\UseRawInputEncoding % because problems with accents in refs

\documentclass[10pt]{article}
\usepackage[letterpaper, left=1in, top=1in, right=1in, bottom=1in, verbose, ignoremp]{geometry}

\usepackage{url}\RequirePackage[colorlinks,citecolor=blue, linkcolor=blue,urlcolor = blue]{hyperref}
\usepackage{latexsym,amssymb,amsmath,amsfonts,graphicx,color,fancyvrb,amsthm,enumerate,subcaption,mathrsfs}
\usepackage[longnamesfirst,authoryear,round]{natbib}
\usepackage[dvipsnames]{xcolor}
\usepackage{xy}\xyoption{all} \xyoption{poly} \xyoption{knot}
\usepackage{float}
\usepackage{bm}
\usepackage{bbm}
\usepackage{multicol,multirow}
\usepackage{array}
\usepackage{relsize}
\usepackage{chngcntr}
\usepackage{etoolbox}
\usepackage{caption}
\usepackage{tikz}
\usetikzlibrary{decorations.pathreplacing,calc}
\usetikzlibrary{shapes,backgrounds}
\usetikzlibrary{patterns}
\usetikzlibrary{cd}
\usepackage{amsopn}
\usepackage{tabularx}
\usepackage{calc}
\usepackage{algorithm}
\usepackage[noend]{algpseudocode}
\usepackage{hyperref}
\usepackage{mathtools}

\usepackage{tikz-cd} % commutative diagrams
\usepackage{amscd} % commutative diagrams
\usepackage{comment}
\usepackage[capitalize,nameinlink,compress]{cleveref}
\crefname{figure}{Figure}{Figures} 
\crefname{equation}{}{} 
\crefname{assumption}{Assumption}{Assumptions}
\crefname{subsection}{Subsection}{Subsections}

\newcounter{cdrow}

\newtheorem{theorem}{Theorem}[]
\newtheorem*{theorem*}{Theorem}

\newtheorem{lemma}[theorem]{Lemma}

\newtheorem*{claim*}{Claim}

\theoremstyle{definition}
\newtheorem{definition}[theorem]{Definition}
\newtheorem*{definition*}{Definition}

\theoremstyle{remark}

\newtheorem*{example*}{Example}

\newcommand*\diff{\mathop{}\!\mathrm{d}}

\makeatletter
\newcommand*{\op}{%
  \DOTSB
  \mathop{\vphantom{\bigoplus}\mathpalette\matt@op\relax}%
  \slimits@
}
\newcommand\matt@op[2]{%
  \vcenter{\m@th\hbox{\resizebox{\widthof{$#1\bigoplus$}}{!}{$\boxplus$}}}%
}
\makeatother

\newcommand{\argmin}{\text{argmin}}

\makeatletter
\def\@biblabel#1{}
\makeatother

\makeatletter
\patchcmd{\NAT@citex}
  {\@citea\NAT@hyper@{%
     \NAT@nmfmt{\NAT@nm}%
     \hyper@natlinkbreak{\NAT@aysep\NAT@spacechar}{\@citeb\@extra@b@citeb}%
     \NAT@date}}
  {\@citea\NAT@nmfmt{\NAT@nm}%
   \NAT@aysep\NAT@spacechar\NAT@hyper@{\NAT@date}}{}{}

\patchcmd{\NAT@citex}
  {\@citea\NAT@hyper@{%
     \NAT@nmfmt{\NAT@nm}%
     \hyper@natlinkbreak{\NAT@spacechar\NAT@@open\if*#1*\else#1\NAT@spacechar\fi}%
       {\@citeb\@extra@b@citeb}%
     \NAT@date}}
  {\@citea\NAT@nmfmt{\NAT@nm}%
   \NAT@spacechar\NAT@@open\if*#1*\else#1\NAT@spacechar\fi\NAT@hyper@{\NAT@date}}
  {}{}

\makeatother

\begin{document}
\def\spacingset#1{\renewcommand{\baselinestretch}%
{#1}\small\normalsize} \spacingset{1}
%\spacingset{1.45} % DON'T change the spacing!

\begin{flushleft}
{\Large{\textbf{$k$-Means Clustering for Persistent Homology}}}
\newline
\\
Yueqi Cao$^{1,\dagger}$, Prudence Leung$^{1,\dagger}$, and Anthea Monod$^{1,\dagger}$
\\
\bigskip
\bf{1} Department of Mathematics, Imperial College London, UK
\\
\bigskip
$\dagger$ Corresponding e-mail: y.cao21@imperial.ac.uk
\end{flushleft}

%%%%%%%%%%%%%%%%%%%%%%%%%%%%%%%%%%%%%%%%%%%%%%%%%%%

\section*{Abstract}
Persistent homology is a methodology central to topological data analysis that extracts and summarizes the topological features within a dataset as a persistence diagram; it has recently gained much popularity from its myriad successful applications to many domains. However, its algebraic construction induces a metric space of persistence diagrams with a highly complex geometry.  In this paper, we prove convergence of the $k$-means clustering algorithm on persistence diagram space and establish theoretical properties of the solution to the optimization problem in the Karush--Kuhn--Tucker framework.  Additionally, we perform numerical experiments on various representations of persistent homology, including embeddings of persistence diagrams as well as diagrams themselves and their generalizations as persistence measures; we find that $k$-means clustering performance directly on persistence diagrams and measures outperform their vectorized representations.

\paragraph{Keywords:} Alexandrov geometry; Karush--Kuhn--Tucker optimization; $k$-means clustering; persistence diagrams; persistent homology.

\section{Introduction}
\label{sec:intro}

Topological data analysis (TDA) is a recently-emerged field of data science that lies at the intersection of pure mathematics and statistics. The core approach of the field is to adapt classical algebraic topology to the computational setting for data analysis.  The roots of this approach first appeared in the early 1990s by \cite{frosini1992measuring} for image analysis and pattern recognition in computer vision \citep{Verri1993}; soon afterwards, these earlier ideas by \cite{frosini_size_2001} were reinterpreted, developed, and extended concurrently by \cite{edelsbrunner_topological_2002} and \cite{zomorodian_computing_2005}. These developments laid the foundations for a comprehensive set of computational topological tools which have then enjoyed great success in applications to a wide range of data settings in recent decades.  Some examples include sensor network coverage problems \citep{de_silva_coverage_2007}; biological and biomedical applications, including understanding the structure of protein data \citep{gameiro_topological_2014,kovacev-nikolic_using_2016,xia_multiscale_2016} and the 3D structure of DNA \citep{emmett_multiscale_2015}, and predicting disease outcome in cancer \citep{crawford2016functional}; robotics and the problem of goal-directed path planning \citep{bhattacharya_persistent_2015}, as well as other general planning and motion problems, such as goal-directed path planning  \citep{doi:10.1177/0278364915586713} and determining the cost of bipedal walking \citep{VASUDEVAN2013101}; materials science to understand the structure of amorphous solids \citep{hiraoka2016hierarchical}; characterizing chemical reactions and chemical structure \citep{doi:10.1021/acs.jctc.2c01204}; financial market prediction \citep{ismail_detecting_2020}; quantifying the structure of music \citep{bergomi2020homological}; among many others \citep[e.g.,][]{otter_roadmap_2017}.

\emph{Persistent homology} is the most popular methodology from TDA due to its intuitive construction and interpretability in applications; it extracts the topological features of data and computes a robust statistical representation of the data under study.  A particular advantage of persistent homology is that it is applicable to very general data structures, such as point clouds, functions, and networks.  Notice that among these examples, the data may be endowed with only a notion of relatedness (point clouds) or may not even be a metric space at all (networks).

A major limitation of persistent homology, however, is that its output takes a somewhat unusual form of a collection of intervals; more specifically, it is a multiset of half-open intervals together with the set of zero-length intervals with infinite multiplicity, known as a \emph{persistence diagram}.  Statistical methods for interval-valued data \citep[e.g.,][]{billard2000regression} cannot be used to analyze persistence diagrams because their algebraic topological construction gives rise to a complex geometry,
which hinders their direct applicability to classical statistical and machine learning methods.  The development of techniques to use persistent homology in statistics and machine learning has been an active area of research in the TDA community.  Largely, the techniques entail either modifying the structure of persistence diagrams to take the form of vectors (referred to as \emph{embedding} or \emph{vectorization}) so that existing statistical and machine learning methodology may be applied directly \citep[e.g.,][]{JMLR:v16:bubenik15a,JMLR:v18:16-337}; or developing new statistical and machine learning methods to accommodate the algebraic structure of persistence diagrams \citep[e.g.,][]{reininghaus_stable_2015}.  The former approach is more widely used in applications, while the latter approach tends to be more technically challenging and requires taking into account the intricate geometry of persistence diagrams and the space of all persistence diagrams.  Our work studies both approaches experimentally and contributes a new theoretical result in the latter setting.  %In the spirit of machine learning, many embedding methods have been proposed for PDs so that existing methods can then be applied to these vectorized representations.  

In machine learning, a fundamental task is clustering, which groups observations of a similar nature to determine structure within the data.  The $k$-means clustering algorithm is one of the most well-known clustering methods. For structured data in linear spaces, the $k$-means algorithm is widely used, with well-established theory on stability and convergence. However, for unstructured data in nonlinear spaces, such as persistence diagrams, it is a nontrivial task to apply the $k$-means framework. The key challenges are twofold: (i) how to define the mean of data in a nonlinear space, and (ii) how to compute the distance between the data points. These challenges are especially significant in the task of obtaining convergence guarantees for the algorithm. For persistence diagrams specifically, the cost function is not differentiable nor subdifferentiable in the traditional sense because of their complicated metric geometry.

%In this paper, we study various representations of PH in the $k$-means clustering algorithm, considered a benchmark clustering method.  Our goal is to understand how much of the original data structure is discernible enough to be accurately groupable when summarized by PH.  In particular, we study the clustering performance on various PD embeddings, on PDs themselves, as well as on their generalizations as persistence measures (PMs).  To implement $k$-means clustering on PD and PM space, we reinterpret the algorithm taking into account the appropriate metrics on these spaces as well as constructions of their respective means.  Experimentally, we show that $k$-means clustering performs better on PDs and PMs than their embeddings via a comprehensive numerical study on both synthetic and real data.  Furthermore, 

In this paper, we study the $k$-means algorithm on the space of persistence diagrams both theoretically and experimentally.  Theoretically, we establish convergence of the $k$-means algorithm.  In particular, we provide a characterization of the solution to the optimization problem in the Karush--Kuhn--Tucker (KKT) framework and show that the solution is a partial optimal point, KKT point, as well as a local minimum.  Experimentally, we study various representations of persistent homology in the $k$-means algorithm, as persistence diagram embeddings, on persistence diagrams themselves, as well as on their generalizations as persistence measures.  To implement $k$-means clustering on persistence diagram and persistence measure space, we reinterpret the algorithm taking into account the appropriate metrics on these spaces as well as constructions of their respective means.  Broadly, our work contributes to the growing body of increasingly important and relevant \emph{non-Euclidean statistics} \citep[e.g.,][]{blanchard2022fr} that deals with, for example, spaces of matrices under nonlinear constraints \citep{10.1214/09-AOAS249}; quotients of Euclidean spaces defined by group actions \citep{le_kume_2000}; and Riemannian manifolds \citep{10.5555/3455716.3455939}.

The remainder of this paper is organized as follows.  We end this introductory section with an overview of existing literature relevant to our study.  Section \ref{sec:preliminaries} provides the background on both persistent homology and $k$-means clustering.  We give details on the various components of persistent homology needed to implement the $k$-means clustering algorithm.  In Section \ref{sec:conv}, we provide main result on the convergence of the $k$-means algorithm in persistence diagram space in the KKT framework.  Section \ref{sec:exp} provides results to numerical experiments comparing the performance of $k$-means clustering on embedded persistence diagrams versus persistence diagrams and persistence measures.  We close with a discussion on our contributions and ideas for future research in Section \ref{sec:end}.

\paragraph{Related Work}  A previous study of the $k$-means clustering in persistence diagram space compares the performance to other classification algorithms; it also establishes local minimality of the solution to the optimization problem in the convergence of the algorithm \citep{maroulas_kmeans_2017}. %\cite{maroulas_kmeans_2017} previously studied the $k$-means clustering algorithm in PD space and compared the performance to other ML classification algorithms; they also establish local minimality of the solution to the optimization problem in convergence of the $k$-means algorithm, but not within a KKT setting. 
%Our work differs by experimentally studying the performance of $k$-means clustering exclusively in the context of PH, while our convergence study expands on the existing study to the more general KKT framework. 
Our convergence study expands on this previous study to the more general KKT framework, which provides a more detailed convergence analysis, and differs experimentally by studying the performance of $k$-means clustering specifically in the context of persistent homology on simulated data as well as a benchmark dataset.

Clustering and TDA more generally has also been recently overviewed, explaining how topological concepts may apply to the problem of clustering \citep{panagopoulos2022topological}.
%\cite{panagopoulos2022topological} overviews clustering and TDA more generally, explaining how topological concepts may apply to the problem of clustering.  
Other work studies topological data analytic clustering for time series data as well as space-time data \citep{islambekov_unsupervised_2019,MAJUMDAR2020113868}.  %\cite{islambekov_unsupervised_2019} created a clustering algorithm based on PH for space-time data.

%%%%%%%%%%%%%%%%%%%%%%%%%%%%%%%%%%%%%%%%%%%%%%%%%%%%%%%%%%%%%%%%%%%%

\section{Background and Preliminaries}
\label{sec:preliminaries}

In this section, we provide details on persistent homology as the context in which our study takes place; we also outline the $k$-means clustering algorithm.  We then provide more technical details on the specifics of persistent homology that are needed to implement the $k$-means clustering algorithm, as well as of Karush--Kuhn--Tucker optimization.

\subsection{Persistent Homology}
\label{sec:ph}

Algebraic topology is a field in pure mathematics that uses abstract algebra to study general topological spaces.  A primary goal of the field is to characterize topological spaces and classify them according to properties that do not change when the space is subjected to smooth deformations, such as stretching, compressing, rotating, and reflecting; these properties are referred to as \emph{invariants}.  Invariants and their associated theory can be seen as relevant to complex and general real-world data settings in the sense that they provide a way to quantify and rigorously describe the ``shape'' and ``size'' of data, which can then be considered as summaries for complex, non-Euclidean data structures and settings are becoming increasingly important and challenging to study.  A particular advantage of TDA, given its algebraic topological roots, is its flexibility and applicability both when there is a clear formalism to the data (e.g., rigorous metric space structure) and when there is a less obvious structure (e.g., networks and text).

Persistent homology adapts the theory of homology from classical algebraic topology to data settings in a dynamic manner.  Homology algebraically identifies and counts invariants as a way to distinguish topological spaces from one another; various different homology theories exist depending on the underlying topological space of study.  When considering the homology of a dataset as a finite topological space $X$, it is common to use \emph{simplicial homology} over a field.  A \emph{simplicial complex} is a reduction of a general topological space to a discrete version, as a union of simpler components, glued together in a combinatorial fashion.  An advantage of considering simplicial complexes and studying simplicial homology is that there are efficient algorithms to compute homology.  More formally, these basic components are \emph{$k$-simplices}, where each $k$-simplex is the convex hull of $k+1$ affinely independent points $x_0, x_1, \ldots, x_k$, denoted by $[x_0,\, x_1,\, \ldots,\, x_k]$.  A set constructed combinatorially of $k$-simplices is a simplicial complex, $K$.

Persistent homology is a dynamic version of homology, which is built on studying homology over a filtration.  A \emph{filtration} is a nested sequence of topological spaces, $X_0 \subseteq X_1 \subseteq \cdots \subseteq X_n = X$; there exist different kinds of filtrations, which are defined by their nesting rule.  In this paper, we study \emph{Vietoris--Rips} (VR) filtrations for finite metric spaces $(X, d_X)$, which are commonly used in applications and real data settings in TDA.  Let $\epsilon_1 \leq \epsilon_2 \leq \cdots \leq \epsilon_n$ be an increasing sequence of parameters where each $\epsilon_i \in \mathbb{R}_{\geq 0}$.  The \emph{Vietoris--Rips complex} of $X$ at scale $\epsilon_i$ $\mathrm{VR}(X, \epsilon_i)$ is constructed by adding a node for each $x_j \in X$ and a $k$-simplex for each set $\{x_{j_1}, x_{j_2}, \ldots, x_{j_{k+1}}\}$ with diameter $d(x_i, x_j)$ smaller than $\epsilon_i$ for all $0 \leq i,j \leq k$.  We then obtain a Vietoris--Rips filtration
$$
    \mathrm{VR}(X,\epsilon_1)\hookrightarrow\mathrm{VR}(X,\epsilon_2)\hookrightarrow\cdots\hookrightarrow\mathrm{VR}(X,\epsilon_n),
$$   
which is a sequence of inclusions of sets.

\begin{figure}
%\vspace{.3in}
\centering
\includegraphics[scale=0.5]{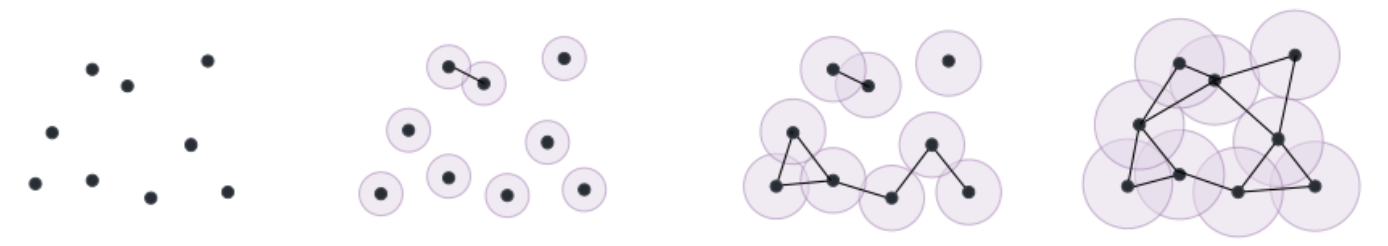}
%\vspace{.3in}
\caption{Example of a VR filtration and VR complexes at various increasing scales on a point cloud.}
\end{figure}

The sequence of inclusions given by the filtration induces maps in homology for any fixed dimension $\bullet$. Let $H_\bullet(X,\epsilon_i)$ be the homology group of $\mathrm{VR}(X,\epsilon_i)$ with coefficients in a field. This gives rise to the following sequence of vector spaces:
$$
    H_\bullet(X,\epsilon_1)\to H_\bullet(X,\epsilon_2)\to\cdots\to H_\bullet(X,\epsilon_n).
$$
The collection of vector spaces $H_\bullet(X,\epsilon_i)$, together with their maps (vector space homomorphisms), $H_\bullet(X,\epsilon_i)\to H_\bullet(X,\epsilon_j)$, is called a \emph{persistence module}. 

The algebraic interpretation of a persistence module and how it corresponds to topological (more specifically, homological) information of data is formally given as follows.  For each finite dimensional $H_\bullet(X,\epsilon_i)$, the persistence module can be decomposed into rank one summands corresponding to birth and death times of homology classes \citep{chazal2016structure}: Let $\alpha\in H_\bullet(X,\epsilon_i)$ be a nontrivial homology class.  $\alpha$ is \emph{born} at $\epsilon_i$ if it is not in the image of $H_\bullet(X,\epsilon_{i-1})\to H_\bullet(X,\epsilon_i)$ and similarly, it \emph{dies} entering $\epsilon_j$ if the image of $\alpha$ via $H_\bullet(X,\epsilon_i)\to H_\bullet(X,\epsilon_{j-1})$ is not in the image $H_\bullet(X,\epsilon_{i-1})\to H_\bullet(X,\epsilon_{j-1})$, but the image of $\alpha$ via $H_\bullet(X,\epsilon_i)\to H_\bullet(X,\epsilon_{j})$ is in the image $H_\bullet(X,\epsilon_{i-1})\to H_\bullet(X,\epsilon_{j})$. The collection of these birth--death intervals $[\epsilon_i,\epsilon_j)$ %is called a \emph{barcode} and it 
represents the \emph{persistent homology} of the Vietoris--Rips filtration of $X$. Taking each interval as an ordered pair of birth--death coordinates and plotting each in a plane $\mathbb{R}^2$ yields a persistence diagram; see Figure \ref{fig:pd}.  This algebraic structure and its construction generate the complex geometry of the space of all persistence diagrams, which must be taken into account when performing statistical analyses on persistence diagrams representing the topological information of data.

More intuitively, homology studies the holes of a topological space as a kind of invariant, since the occurrence of holes and how many there are do not change under continuous deformations of the space.  Dimension 0 homology corresponds to connected components; dimension 1 homology corresponds to loops; and dimension 2 homology corresponds to voids or bubbles.  These ideas can be generalized to higher dimensions and the holes are topological features of the space (or dataset).  Persistent homology tracks the appearance of these features with an interval, where the left endpoint of the interval signifies the first appearance of a feature within the filtration and the right endpoint signifies its disappearance.  The length of the interval can therefore be thought of as the lifetime of the topological feature.  The collection of the intervals is known as a \emph{barcode} and summarizes the topological information of the dataset (according to dimension).  A barcode can be equivalently represented as a persistence diagram, which is a scatterplot of the birth and death points of each interval. See Figure \ref{fig:illustration} for a visual illustration.

\begin{figure}
\begin{subfigure}{0.45\textwidth}
\centering
\includegraphics[width=\textwidth]{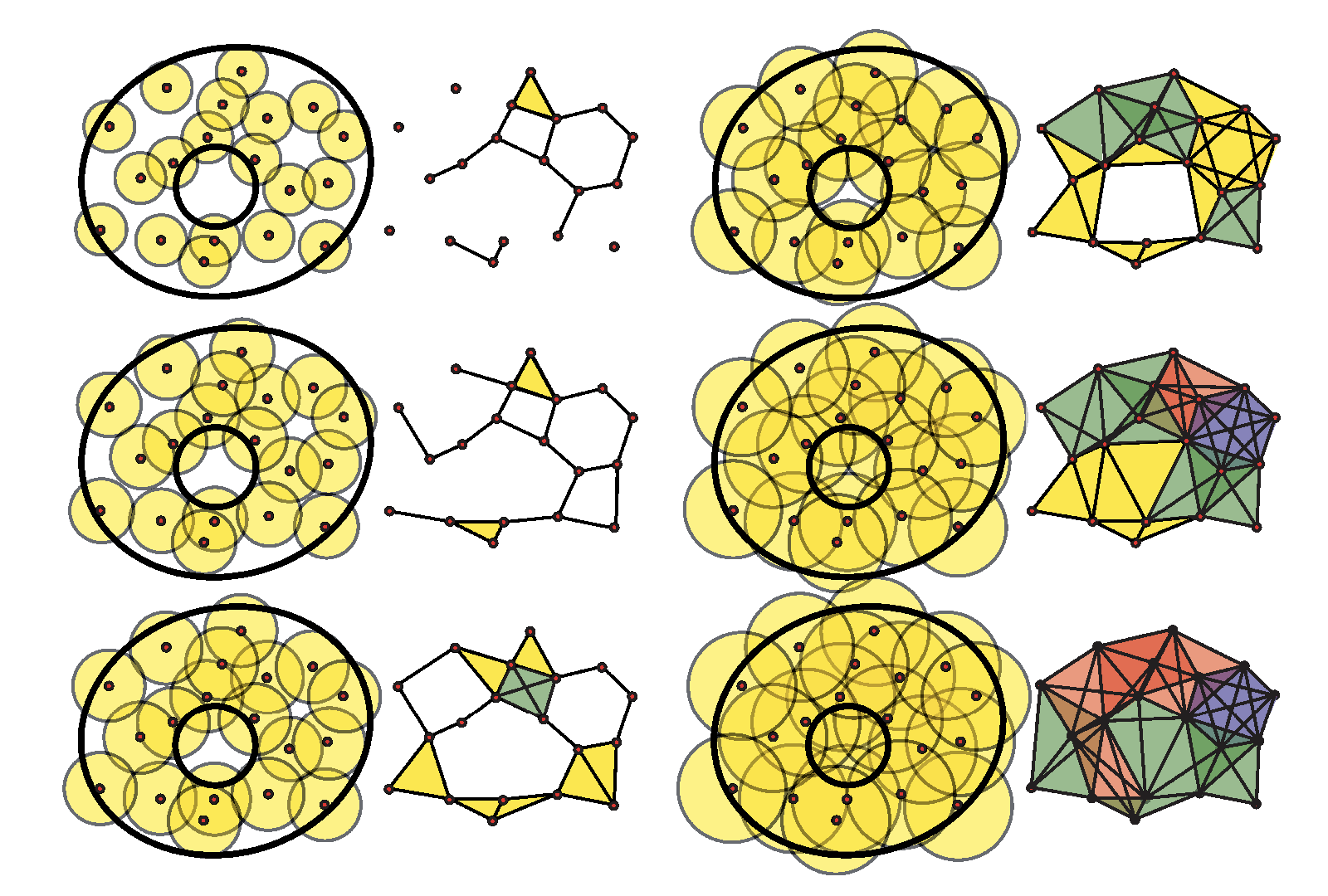}
\end{subfigure}
\begin{subfigure}{0.45\textwidth}
\centering
\includegraphics[width=\textwidth]{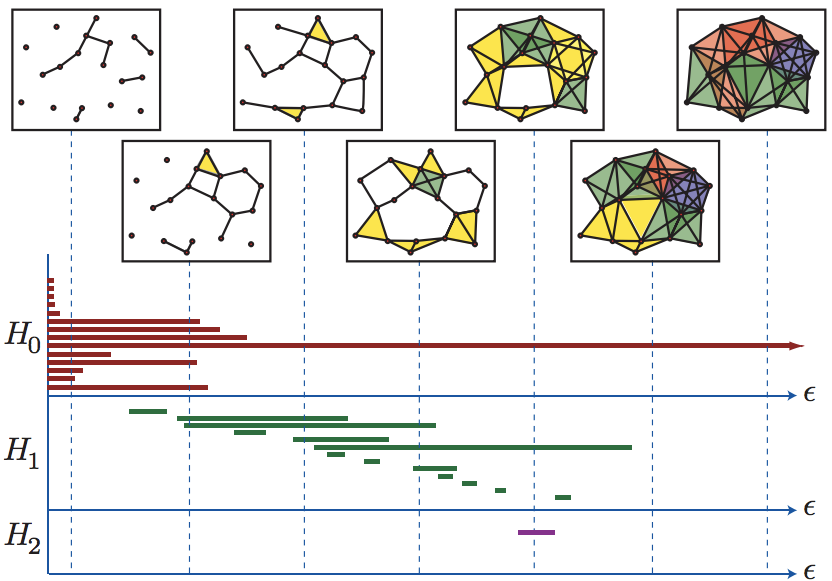}
\end{subfigure}
\caption{Example illustrating the procedure of persistent homology \citep{ghrist2008barcodes}.  The panel on the left shows an annulus, which is the underlying space of interest; the topology of an annulus is characterized by 1 connected component (corresponding to $H_0$), 1 loop (corresponding to $H_1$), and 0 voids (corresponding to $H_2$).  17 points are sampled from the annulus and a simplicial complex (in order to be able to use simplicial homology) is constructed from these samples in the following way: The sampled points are taken to be centers of balls with radius $\epsilon$, which grows from 0 to $\infty$.  At each value of $\epsilon$, whenever two balls intersect, they are connected by an edge; when three balls mutually intersect, they are connected by a triangle (face); this construction continues for higher dimensions.  The simplicial complex resulting from this overlapping of balls as the radius $\epsilon$ grows is shown to the right of each snapshot at different values of $\epsilon$.  As $\epsilon$ grows, connected components merge, cycles are formed, and fill up.  These events are tracked by a bar for each dimension of homology, shown on the right; the lengths of the bars corresponds to the lifetime of each topological feature as $\epsilon$ grows.  Notice that as $\epsilon$ reaches infinity, there remains a single connected component tracked by the longest bar for $H_0$.  For $H_1$, there are bars of varying length, including several longer bars, which suggests irregular sampling of the annulus; the longest bar corresponds to the single loop of the annulus.  Notice also that there is one short $H_2$ bar, which is likely to be a spurious topological feature corresponding to noisy sampling.}
\label{fig:illustration}
\end{figure}

\begin{definition}
A {\em persistence diagram} $D$ is a locally finite multiset of points in the half-plane $\Omega = \{(x,y)\in\mathbb{R}^2 \mid x<y\}$ together with points on the diagonal $\partial\Omega=\{(x,x)\in\mathbb{R}^2\}$ counted with infinite multiplicity. Points in $\Omega$ are called {\em off-diagonal points}. The persistence diagram with no off-diagonal points is called the {\em empty persistence diagram}, denoted by $D_\emptyset$. 
%We denote the space of all persistence diagrams by $\mathcal{D}_p$.
\end{definition}

Intuitively, points that are far away from the diagonal are the long intervals corresponding to topological features that have a long lifetime in the filtration.  These points are said to have \emph{high persistence} and are generally considered to be topological signal, while points closer to the diagonal are considered to be topological noise.

%\paragraph{Representing Persistent Homology.}
\paragraph{Persistence Measures}

A more general representation of persistence diagrams that aims to facilitate statistical analyses and computations is the \emph{persistence measure}.  The construction of persistence measures arises from the equivalent representation of persistence diagrams as measures on $\Omega$ of the form $\sum_{x\in D\cap\Omega}n_x\delta_x$, where $x$ ranges over the off-diagonal points in a persistence diagram $D$, $n_x$ is the multiplicity of $x$, and $\delta_x$ is the Dirac measure at $x$ \citep{divol2019density}.  Considering then all Radon measures supported on $\Omega$ gives the following definition.

\begin{definition}[\cite{divol_understanding_2020}]
\label{def:pers_meas}
Let $\mu$ be a Radon measure supported on $\Omega$. For $1\le p<\infty$, the \emph{$p$-total persistence} in this measure setting is defined as 
$$
    \mathrm{Pers}_p(\mu) = \int_{\Omega} \|x-x^\top\|_q^p\diff{\mu(x)},
$$
where $x^\top$ is the projection of $x$ to the diagonal $\partial\Omega$. Any Radon measure with finite $p$-total persistence is called a {\em persistence measure}. 
\end{definition}

Persistence measures resemble heat map versions of persistence diagrams on the upper half-plane where individual persistence points are smoothed, so the intuition on their birth--death lifetimes as described previously and illustrated in Figure \ref{fig:illustration} are lost.  However, they gain the advantage of a linear structure where statistical quantities are well-defined and much more straightforward to compute \citep{divol-2021-estimation}.

By generalizing persistence diagrams to persistence measures, a larger space of persistence measures is obtained, which contains the space of persistence diagrams as a proper subspace \citep{divol_understanding_2020}. In other words, when it comes to discrete measures on the half plane, a persistence measure is nothing but a persistence diagram with different representation.  Thus, any distribution on the space of persistence diagrams is naturally defined on the space of persistence measures. The mean of distributions is also well defined for persistence measures if it is well defined for persistence diagrams (see Section \ref{sec:means}). Note that on the space of persistence measures, the mean of a distribution is a persistence measure, which can be an abstract Radon measure that is not representable as a persistence diagram. In fact, in general, the empirical mean on the space of persistence measures does not lie in the subspace of persistence diagrams. However, it enjoys the computational simplicity which does not happen to other means defined on the space of persistence diagrams.

\paragraph{Metrics for Persistent Homology}

There exist several metrics for persistence diagrams; we focus on the following for its computational feasibility and the known geometric properties of the set of all persistence diagrams (persistence diagram space) induced by this metric.

\begin{definition}
\label{def:wass_dist}
For $p > 0$, the {\em $p$-Wasserstein distance} between two persistence diagrams $D_1$ and $D_2$ is defined as
$$
\mathrm{W}_{p,q}(D_1, D_2) = \inf_{\gamma:D_1 \rightarrow D_2} \Bigg( \sum_{x \in D_1}\|x - \gamma(x) \|_q^p \Bigg)^{1/p},
$$
where $\|\cdot\|_q$ denotes the $q$-norm, $1\leq q \leq \infty$ and $\gamma$ ranges over all bijections between $D_1$ and $D_2$.
\end{definition}

For a persistence diagram $D$, its \emph{$p$-total persistence} is $\mathrm{W}_{p,q}(D, D_\emptyset)$. Notice here that the intuition is similar to that given in Definition \ref{def:pers_meas}, but technically different to the Radon measure setting. We refer to the \emph{space of persistence diagrams} $\mathcal{D}_p$ as the set of all persistence diagrams with finite $p$-total persistence. Equipped with the $p$-Wasserstein distance, $(\mathcal{D}_p, \mathrm{W}_{p,q})$ is a well-defined metric space.

The space of persistence measures is equipped with the following metric.

\begin{definition}
\label{def:opt_dist}
The \emph{optimal partial transport distance} between equivalent representation persistence measures $\mu$ and $\nu$ is 
$$
\mathrm{OT}_{p,q}(\mu, \nu) = \inf_\Pi \Bigg(\int_{\bar{\Omega} \times \bar{\Omega}} \|x - y \|_q^p \diff \Pi(x, y) \Bigg)^{1/p},
$$
with $\bar{\Omega} = \Omega \cup \partial\Omega$. Here, $\Pi$ denotes the set of all Radon measures on $\bar{\Omega} \times \bar{\Omega}$ where for all Borel sets $A, B \subseteq \Omega$, $\Pi(A \times \bar{\Omega}) = \mu(A)$ and $\Pi(B \times \bar{\Omega}) = \nu(B)$. 
\end{definition}

The space of persistence measures $\mathcal{M}_p$, together with $\mathrm{OT}_{p, q}$, is an extension and generalization of $(\mathcal{D}_p, \mathrm{W}_{p,q})$.  In particular, $(\mathcal{D}_p, \mathrm{W}_{p,q})$ is a proper subset of $(\mathcal{M}_p, \mathrm{OT}_{p,q})$ and $\mathrm{OT}_{p,q}$ coincides with $\mathrm{W}_{p,q}$ on the subset $\mathcal{D}_p$ \citep{divol-2021-estimation}.

\subsection{$k$-Means Clustering}

The $k$-means clustering algorithm is one of the most widely known unsupervised learning algorithms.  It groups data points into $k$ distinct clusters based on their similarities and without the need to provide any labels for the data \citep{hartigan_algorithm_1979}. The process begins with a random sample of $k$ points from the set of data points. Each data point is assigned to one of the $k$ points, called \emph{centroids}, based on minimum distance. Data points assigned to the same centroid form a \emph{cluster}. The cluster centroid is updated to be the mean of the cluster. Cluster assignment and centroid update steps are repeated until convergence. 

\begin{algorithm}

  \caption{$k$-Means Algorithm for Persistence Diagrams/Measures}
   \label{code:KM-alg}
   \begin{algorithmic}[1]
   \Require{$n$ persistence diagrams ($D_1, \ldots, D_n$), fixed integer $k$.}
   \Statex
   \State \textbf{Initialization}: From ($D_1, \ldots, D_n$), randomly select $k$ points as the initial centroids ($m_1, \ldots, m_k$).
   \State \textbf{Assignment}:  Assign each data point $D_i$ to cluster $k_i^*$ based on distance to each of the $k$ centroids:
   $$
   k_i^* = \argmin_j \operatorname{dist}(D_i, m_j)
   $$
   where $\operatorname{dist}(\cdot\,,\,\cdot)$ stands for $p$-Wasserstein distance for persistence diagrams and $p$-optimal partial transport distance for persistence measures.
    \State \textbf{Update}: Recompute the Fr\'{e}chet mean or mean persistence measure $m_k$ as the mean of the data points belonging to each cluster.
\Statex
\Statex
Repeat assignment and update steps until cluster assignment does not change, i.e., the algorithm converges.
\end{algorithmic}

\end{algorithm}

A significant limitation of the $k$-means algorithm is that the final clustering output is highly dependent on the initial choice of $k$ centroids. If initial centroids are close in terms of distance, this may result in suboptimal partitioning of the data. Thus, when using the $k$-means algorithm, it is preferable to begin with centroids that are as far from each other as possible such that they already lie in different clusters, which makes it more likely for convergence to the global optimum rather than a local solution. $k$-means$++$ \citep{arthur_k-means_2006} is an initialization algorithm that addresses this issue which we adapt in our implementation of the $k$-means algorithm, in place of the original random initialization step.

Finding the true number of clusters $k$ is difficult in theory. Practically, however, there are several heuristic ways to determine the near optimal number of clusters. For example, the \emph{elbow method} seeks the cutoff point of the loss curve \citep{thorndike1953belongs}; the \emph{information criterion method} constructs likelihood functions for the clustering model \citep{goutte2001feature}; the \emph{silhouette method} tries to optimize the silhouette function of data with respect to the number of clusters \citep{de2015recovering}. An experimental comparison of different determination rules can be found in \cite{pham2005selection}.

The original $k$-means algorithm \cite{hartigan_algorithm_1979} takes vectors as input and calculates the squared Euclidean distance as the distance metric in the cluster assignment step.  In this work, we aim to adapt this algorithm that intakes representations of persistent homology as input (persistence diagrams, persistence measures, and embedded persistence diagrams) and the Wasserstein distance (Definition \ref{def:wass_dist}) and the optimal partial transport distance (Definition \ref{def:opt_dist}) between persistence diagrams and persistence measures, respectively, as metrics.  %We adapt these inputs for the three embeddings of persistence diagrams outlined previously. For persistence diagrams and persistence measures, we adapt the algorithm to intake the $p$-Wasserstein distance $W_{p,q}$ and optimal partial transport distance $\mathrm{OT}_{p,q}$, respectively. 

%In the following section we present the notions of centrality for each of the persistent homology representations. These will serve as the means calculated in the update step of the $k$-means algorithm.

\subsection{Means for Persistent Homology}
\label{sec:means}

The existence of a mean and the ability to compute it is fundamental in the $k$-means algorithm.  This quantity and its computation are nontrivial in persistence diagram space.  %We present notions of centrality for the PH representations discussed above, allowing for proper implementation of the $k$-means algorithm. 
A previously proposed notion of means for sets of persistence diagrams is based on \emph{Fr\'{e}chet means}---generalizations of means to arbitrary metric spaces---where the metric is the $p$-Wasserstein distance \cite{turner_frechet_2013}.

Specifically, given a set of persistence diagrams $\mathbf{D} = \{D_1, \ldots, D_N\}$, define an empirical probability measure $\hat{\rho} = \frac{1}{N} \sum_{i=1}^N \delta_{D_i}$. The \textit{Fr\'{e}chet function} for any persistence diagram $D \in \mathbf{D}$ is then
\begin{equation}
\label{eq:frechet_function}
\mathrm{F}_{\hat{\rho}}(D) = \frac{1}{N}\sum_{i=1}^N \mathrm{W}_{p, p}(D, D_i).
\end{equation}
Since any persistence diagram $D$ that minimizes the Fr\'{e}chet function is the Fr\'{e}chet mean of the set $\mathbf{D}$, the Fr\'{e}chet mean in general need not be unique.  This is the case in $(\mathcal{D}_2, W_{2,2})$, which is known to be an Alexandrov space with curvature bounded from below, meaning that geodesics and therefore Fr\'{e}chet means are not even locally unique \citep{turner_frechet_2013}. Nevertheless, they are computable as local minima, using a greedy algorithm that is a variant of the Hungarian algorithm, but have caveats: the Fr\'{e}chet function is not convex on $\mathcal{D}_2$, which means the algorithm often finds only local minima; additionally, there is no initialization rule and arbitrary persistence diagrams are chosen as starting points, meaning that different initializations lead to an unstable performance and potentially different local minima.  In our work, we will focus on the 2-Wasserstein distance and we use the notation $\mathrm{W}_2$ for short.

%Under the Alexandrov geometry of $(\mathcal{D}_2, W_2)$,  Fréchet means of PDs can be computed using a variant of the Hungarian algorithm proposed by \cite{turner_frechet_2013}. The non-convexity of the algorithm means that there is no guarantee of finding the global optimum and the algorithm is likely to converge to a local minima.

For a set of persistence measures $\{\mu_1, \ldots, \mu_N\}$, their mean is the \emph{mean persistence measure}, which is simply the empirical mean of the set: $\bar{\mu} = \frac{1}{N}\sum_{i=1}^N\mu_i$. 

We outline the $k$-means algorithm for persistence diagrams and persistence measures in Algorithm \ref{code:KM-alg}.

\subsection{The Karush--Kuhn--Tucker Conditions}
\label{sec:kkt}

The Karush--Kuhn--Tucker (KKT) conditions are first order necessary conditions for a solution in a nonlinear optimization problem to be optimal. 

\begin{definition}
Consider the following minimization problem, 
\begin{align*}
    \min_{x\in\mathbb{R}^d}\quad f(x) \quad
    \text{s.t.}\quad g_i(x) & = 0,\,i=1,\ldots,I\\
    h_j(x) &\leq 0, \,j=1,\ldots,J
\end{align*}
where $f$, $g_i$, and $h_j$ are differentiable functions on $\mathbb{R}^d$. The \emph{Lagrange dual function} is given by
$$
g(u,v) =  \min_{x\in\mathbb{R}^d} f(x)+\sum_{i\in I}u_ig_i(x)+\sum_{j\in J}v_jg_j(x).
$$
The \emph{dual problem} is defined as
$$
    \max_{u\in\mathbb{R}^I,v\in\mathbb{R}^J}\quad g(u,v) \quad
    \text{s.t.}\quad v_j\geq 0,\, j=1,\ldots,J.
$$
The \emph{Karush--Kuhn--Tucker (KKT) conditions} are given by
$$
\begin{cases}
     \nabla f(x) + \sum_{i\in I}u_i\nabla g_i(x) + \sum_{j\in J}v_j\nabla h_j(x) = 0\\
     v_jh_j(x) = 0,\, & j=1,\ldots,J\\
     v_j\geq 0,\, & j=1,\ldots,J\\
     g_i(x) = 0,\, h_j(x)\leq 0,\, & i=1,\ldots,I,\,j=1,\ldots,J.
\end{cases}
$$
\end{definition}

Given any feasible $x$ and $u,v$, the difference $f(x)-g(u,v)$ is called the \emph{dual gap}. The KKT conditions are always sufficient conditions for solutions to be optimal; they are also necessary conditions if strong duality holds, i.e., if the dual gap is equal to zero \citep{boyd2004convex}.

The KKT conditions are applicable if all functions are differentiable over Euclidean spaces. The setup can be generalized for non-differentiable functions over Euclidean spaces but with well-defined subdifferentials \citep{boyd2004convex}. In the case of persistent homology, the objective function is defined on the space of persistence diagrams, which is a metric space without linear structure and therefore requires special treatment.  Specifically, we follow \cite{selim1984k} and define the \emph{KKT point} for our optimization problem and show that the \emph{partial optimal points} are always KKT points.  These notions are formally defined and full technical details of our contributions are given in the next section.

%%%%%%%%%%%%%%%%%%%%%%%%%%%%%%%%%%%%%%%%%%%%%%%%

\section{Convergence of the $k$-Means Clustering Algorithm}
\label{sec:conv}

Local minimality of the solution to the optimization problem has been previously established by studying the convergence of the $k$-means algorithm for persistence diagrams \citep{maroulas_kmeans_2017}.  Here, we further investigate the various facets of the optimization problem in a KKT framework and study and derive properties of the solution taking into account the subtleties and refinements of the KKT setting previously discussed in Section \ref{sec:kkt}.  We focus on the 2-Wasserstein distance $\mathrm{W}_2$ as discussed in Section \ref{sec:means}. %where, in particular, the square distance function admits gradients in the Alexandrov sense \citep{turner_frechet_2013}.

%In the following we will focus on the 2-Wasserstein distance on the space of persistence diagrams, since in such distance the space is an Alexandrov space with curvature bounded below and the square distance function admits gradients in the Alexandrov sense \citep{turner_frechet_2013}.  

%\subsection{Problem Setup}
\paragraph{Problem Setup}  

Let $\mathbf{D} = \{D_1,\ldots,D_n\}$ be a set of $n$ persistence diagrams and let $\mathbf{G}=\{G_1,\ldots,G_k\}$ be a partition of $\mathbf{D}$, i.e., $G_i\cap G_j=\emptyset$ for $i\neq j$ and $\cup_{i=1}^kG_i=\mathbf{D}$. Let $\mathbf{Z}=\{Z_1,\ldots,Z_k\}$ be $k$ persistence diagrams representing the centroids of a clustering. The within-cluster sum of distances is defined as
\begin{equation}\label{eq:cost-original}
    \mathbf{F}(\mathbf{G},\mathbf{Z}) = \sum_{i=1}^k\sum_{D_j\in G_i}\mathrm{W}_2^2(D_j,Z_i).
\end{equation}
Define a $k\times n$ matrix $\omega=[\omega_{ij}]$ by setting $\omega_{ij} = 1$ if $D_j \in G_i$ and 0 otherwise.
%\begin{equation*}
%    \omega_{ij} = \begin{cases}1\text{ if }D_j\in G_i;\\ 0 \text{ otherwise.}\end{cases}
%\end{equation*}
Then \eqref{eq:cost-original} can be equivalently expressed as
%\begin{equation*}
$$ 
      \mathbf{F}(\omega,\mathbf{Z}) = \sum_{i=1}^k\sum_{j=1}^n\omega_{ij}\mathrm{W}_2^2(D_j,Z_i).
$$
%\end{equation*}

We first define a domain for each centroid $Z_i$: Let $\ell_i$ be the number of off-diagonal points of $D_i$; let $V$ be the convex hull of all off-diagonal points in $\cup_{i=1}^n D_i$. Define $\mathbf{S}$ to be the set of persistence diagrams with at most $\sum_{i=1}^n\ell_i$ off-diagonal points in $V$; the set $\mathbf{S}$ is relatively compact \citep[][Theorem 21]{mileyko_probability_2011}.

Our \emph{primal} optimization problem is 
\begin{align*}
%\begin{equation}
\label{eq:primal}
    \min\quad & \mathbf{F}(\omega,\mathbf{Z})\notag\\
    \text{s.t.}\quad & \begin{cases}\sum_{i=1}^k \omega_{ij} =1 & \forall\ j=1,\ldots,n\\
    \omega_{ij} = 0\text{ or }1 & \\
    Z_i\in\mathbf{S} & \forall\ i=1,\ldots,k.\end{cases}\tag{$\mathbf{P}_0$}
%\end{equation}
\end{align*}
We relax the integer constraints in optimization \eqref{eq:primal} to continuous constraints by setting $\omega_{ij}\in[0,1]$. Let $\Theta$ be the set of all $k\times n$ column stochastic matrices, i.e.,
$$
\Theta = \{\omega\in\mathbb{R}^{k\times n}\mid \omega_{ij}\ge 0,\text{ and } \sum_{i}w_{ij}=1 \ \forall\ j\}.
$$
The \emph{relaxed} optimization problem is then
%\begin{align*}
%    \min\quad & \mathbf{F}(\omega,\mathbf{Z})\tag{$\mathbf{P}$}\\
%    s.t.\quad & \omega\in\Omega,\, \mathbf{Z}\in\mathbf{S}^k
%\end{align*}
%$$
\begin{equation}
\label{eq:relax}
\min\quad  \mathbf{F}(\omega,\mathbf{Z}) \quad \text{s.t.} \quad  \omega\in\Theta,\, \mathbf{Z}\in\mathbf{S}^k. \tag{$\mathbf{P}$}
\end{equation}
%$$
Finally, we reduce $\mathbf{F}(\omega,\mathbf{Z})$ to a univariate function $f(\omega)=\min\{F(\omega,\mathbf{Z}),\mathbf{Z}\in\mathbf{S}^k\}$ by minimizing over the second variable. The \emph{reduced} optimization problem is 
%\begin{align*}
%    \min\quad & f(\omega)\tag{$\mathbf{RP}$}\\
%    s.t.\quad & \omega\in\Omega
%\end{align*}
%$$
\begin{equation}
\label{eq:reduced}
\min\quad f(\omega) \quad \text{s.t.} \quad \omega\in\Theta. \tag{$\mathbf{RP}$}
\end{equation}
%$$

\begin{lemma}
The optimization problems \eqref{eq:primal}, \eqref{eq:relax}, and \eqref{eq:reduced} are equivalent in the sense that the minimum values of the cost functions are equal.
\end{lemma}
\begin{proof}
Problems \eqref{eq:relax} and \eqref{eq:reduced} are equivalent as they have the same constraints and equal cost functions. It suffices to prove that the reduced problem \eqref{eq:reduced} and the primal problem \eqref{eq:primal} are equivalent.

\noindent
\emph{Claim 1}: $f(w)$ is a concave function on the convex set $\Theta$.

Let $\omega_1,\omega_2\in\Theta$, and $\lambda\in[0,1]$. Set $\omega_\lambda=\lambda\omega_1+(1-\lambda)\omega_2$. Note that $\sum_{i}(\omega_\lambda)_{ij} = \lambda\sum_{i}(\omega_1)_{ij}+(1-\lambda)\sum_i(\omega_2)_{ij}=1$. Thus $\omega_\lambda\in \Theta$. Furthermore,
%\begin{equation*}
%\begin{aligned}
\begin{align*}
f(\omega_\lambda) &= \min_{\mathbf{Z}}\{\mathbf{F}(\omega_\lambda,\mathbf{Z})\}\\
&= \min_{\mathbf{Z}}\{\lambda\mathbf{F}(\omega_1,\mathbf{Z})+(1-\lambda)\mathbf{F}(\omega_2,\mathbf{Z})\}\\
&\ge \lambda\min_{\mathbf{Z}}\{\mathbf{F}(\omega_1,\mathbf{Z})\}+(1-\lambda)\min_{\mathbf{Z}}\{\mathbf{F}(\omega_2,\mathbf{Z})\}\\
&= \lambda f(\omega_1)+(1-\lambda)f(\omega_2).
\end{align*}
%\end{aligned}
%\end{equation*}
As a consequence, $f(\omega)$ always attains its minimum on the extreme points of $\Theta$.  Note that an extreme point is a point that cannot lie within the open interval of the convex combination of any two points.

\noindent
\emph{Claim 2}: The extreme points of $\Theta$ satisfy the constraints in \eqref{eq:primal}.

The extreme points of $\Theta$ are 0--1 matrices which have exactly one 1 in each column \citep{cao2022centrosymmetric}. They are precisely the constraints in \eqref{eq:primal}.

In summary, any optimal point for \eqref{eq:reduced} is also an optimal point for \eqref{eq:primal} and vice versa. Thus, all three optimization problems are equivalent. 
\end{proof}

\subsection{Convergence to Partial Optimal Points}

We prove that Algorithm \ref{code:KM-alg} for persistence diagrams always converges to partial optimal points in a finite number of steps. The proof is an extension of Theorem 3.2 in \cite{maroulas_kmeans_2017}.

\begin{definition}{\citep{selim1984k}}
A point $(\omega_\star,\mathbf{Z}_\star)$ is a \emph{partial optimal point} of \eqref{eq:relax} if 
\begin{enumerate}[(i)]
\item $\mathbf{F}(\omega_\star,\mathbf{Z}_\star)\le \mathbf{F}(\omega,\mathbf{Z}_\star)$ for all $\omega\in\Theta$, i.e., $\omega_\star$ is a minimizer of the function $\mathbf{F}(\cdot,\mathbf{Z}_\star)$; and 

\item $\mathbf{F}(\omega_\star,\mathbf{Z}_\star)\le \mathbf{F}(\omega_\star,\mathbf{Z})$ for all $\mathbf{Z}\in\mathbf{S}^k$, i.e., $\mathbf{Z}_\star$ is a minimizer of the function $\mathbf{F}(\omega_\star,\cdot)$.
\end{enumerate}
%\begin{enumerate}
%    \item $\mathbf{F}(\omega_\star,\mathbf{Z}_\star)\le \mathbf{F}(\omega,\mathbf{Z}_\star)$ for all $\omega\in\Omega$, i.e. $\omega_\star$ is a minimizer of the function $\mathbf{F}(\cdot,\mathbf{Z}_\star)$;
%    \item $\mathbf{F}(\omega_\star,\mathbf{Z}_\star)\le \mathbf{F}(\omega_\star,\mathbf{Z})$ for all $\mathbf{Z}\in\mathbf{S}^k$, i.e. $\mathbf{Z}_\star$ is a minimizer of the function $\mathbf{F}(\omega_\star,\cdot)$.
%\end{enumerate}
\end{definition}

\begin{theorem}{\citep[][Theorem 3.2]{maroulas_kmeans_2017}}
\label{thm:converge-pos}
The $k$-means algorithm over $(\mathcal{D}_2, \mathrm{W}_{2})$ converges to a partial optimal point in finitely many steps.
\end{theorem}
\begin{proof}
The value of $\mathbf{F}$ only changes at two steps during each iteration. Fix an iteration $t$ and let $\mathbf{Z}^{(t)} = \{Z_1^{(t)},\ldots,Z_k^{(t)}\}$ be the $k$ centroids from iteration $t$. At the first step of the $(t+1)$st iteration, since we are assigning all persistence diagrams to a closest centroid, for any datum $D_j$,
\begin{align*}
\sum_{i=1}^k\omega_{ij}^{(t+1)}\mathrm{W}_2^2(D_j,Z_i^{(t)})=&\min_i \{\mathrm{W}_2^2(D_j,Z_i^{(t)})\}\\
=&\sum_{i=1}^k\omega_{ij}^{(t)}\min_i \{\mathrm{W}_2^2(D_j,Z_i^{(t)})\}\\
\le&\sum_{i=1}^k\omega_{ij}^{(t)}\mathrm{W}_2^2(D_j,Z_i^{(t)}).
\end{align*}
% $$
% \sum_{i=1}^k\omega_{ij}^{(t+1)}\mathrm{W}_2^2(D_j,Z_i^{(t)})=\min_i \{\mathrm{W}_2^2(D_j,Z_i^{(t)})\}
% =\sum_{i=1}^k\omega_{ij}^{(t)}\min_i \{\mathrm{W}_2^2(D_j,Z_i^{(t)})\}
% \le\sum_{i=1}^k\omega_{ij}^{(t)}\mathrm{W}_2^2(D_j,Z_i^{(t)}).
% $$
Summing over $j=1,\ldots,n$, we have $\mathbf{F}(\omega^{(t+1)},\mathbf{Z}^{(t)})\le \mathbf{F}(\omega^{(t)},\mathbf{Z}^{(t)})$.

At the second step of the $(t+1)$st iteration, by the definition of Fr\'{e}chet means, we have
$$
Z_i^{(t+1)} = \mathop{\arg\min}_Z \sum_{j=1}^n\omega_{ij}^{(t+1)}\mathrm{W}_2^2(D_j,Z).
$$
Thus
$$
\sum_{j=1}^n\omega_{ij}^{(t+1)}\mathrm{W}_2^2(D_j,Z_i^{(t+1)})\le \sum_{j=1}^n\omega_{ij}^{(t+1)}\mathrm{W}_2^2(D_j,Z_i^{(t)}).
$$
Summing over $i=1,\ldots,k$, we have $\mathbf{F}(\omega^{(t+1)},\mathbf{Z}^{(t+1)})\le\mathbf{F}(\omega^{(t+1)},\mathbf{Z}^{(t)})$. 

In summary, the function $\mathbf{F}(\omega,\mathbf{Z})$ is nonincreasing at each iteration. Since there are only finitely many extreme points in $\Theta$ \citep{cao2022centrosymmetric}, the algorithm is stable at a point $\omega_\star$ after finitely many steps. If we fix one Fr\'{e}chet mean $\mathbf{Z}_\star$ for $\mathbf{F}(\omega_\star,\cdot)$, then $(\omega_\star,\mathbf{Z}_\star)$ is a partial optimal point for \eqref{eq:relax} by construction.
\end{proof}
%\begin{remark}
Theorem \ref{thm:converge-pos} is a revision of the original statement in \cite{maroulas_kmeans_2017}, since the algorithm may not converge to local minima of \eqref{eq:relax}. %We will discuss the convergence to local minima in the next subsection.
%\end{remark}

Note that $\mathbf{F}(\omega,\mathbf{Z})$ is differentiable as a function of $\omega$ but not differentiable as a function of $\mathbf{Z}$. %However, under the 2-Wasserstein distance the space of persistence diagrams $(\mathcal{P},\mathrm{W}_2)$ is an Alexandrove space with non-negative curvature, and it is possible to define a differential structure on $\mathcal{P}$ as in \cite{turner_frechet_2013}. \textcolor{red}{Is it possible to put the metric geometry background to appendix or supplementary materials?} 
However, given that we are restricting to $(\mathcal{D}_2, \mathrm{W}^2_2)$, there is a differential structure and in particular, the square distance function admits gradients in the Alexandrov sense \citep{turner_frechet_2013}.

Given a point $D\in \mathcal{D}_2$, let $\Sigma_D$ be the set of all nontrivial unit-speed geodesics emanating from $D$.  The angle between two geodesics $\gamma_1,\gamma_2$ is defined as
$$
\angle_D(\gamma_1,\gamma_2) = \arccos\bigg(\lim_{s,t\to 0}\frac{s^2+t^2-\mathrm{W}_2(\gamma_1(s),\gamma_2(t))}{2st}\bigg).
$$
By identifying geodesics $\gamma_1\sim\gamma_2$ with zero angles, we define the space of directions $(\widehat{\Sigma}_D,\angle_D)$ as the completion of $\Sigma_D/\sim$ under the angular metric $\angle_D$. The tangent cone $T_D=\widehat{\Sigma}_D\times[0,\infty)/\widehat{\Sigma}_D\times\{0\}$ is the Euclidean cone over $\widehat{\Sigma}_D$ with the cone metric defined as
$$
\mathrm{C}_D([\gamma_1,t],[\gamma_2,s])=s^2+t^2-2st\cos\angle_D(\gamma_1,\gamma_2).
$$
The inner product is defined as
$
\langle[\gamma_1,t],\, [\gamma_2,s] \rangle_D = st\cos\angle_D(\gamma_1,\gamma_2).
$
Let $v=[\gamma,t]\in T_D$. For any function $h:\mathcal{P}\to\mathbb{R}$, the differential of $h$ at a point $D$ is a map $T_D\to\mathbb{R}$ defined by
$$
\mathrm{d}_Dh(v) = \lim_{s\to 0} \frac{h(\gamma_v(s))-h(D)}{s}
$$
if the limit exists and is independent of selection of $\gamma_v$. The gradient of $h$ at $D$ is a tangent vector $\nabla_Dh\in T_D$ such that (i) $\mathrm{d}_Dh(u)\le \langle \nabla_Dh,u \rangle_D$ for all $u\in T_D$; and (ii) $\mathrm{d}_Dh(\nabla_Dh) = \langle\nabla_Dh,\nabla_Dh\rangle$.
%\begin{enumerate}
%    \item  $\mathrm{d}_Dh(u)\le \langle \nabla_Dh,u \rangle_D$ for all $u\in T_D$;
%    \item $\mathrm{d}_Dh(\nabla_Dh) = \langle\nabla_Dh,\nabla_Dh\rangle$
%\end{enumerate}
Let $Q_D(\cdot)=\mathrm{W}_2^2(\cdot,D)$ be the squared distance to $D$. %We list some important facts about $Q_D$ from \cite{turner_frechet_2013}.
\begin{lemma}{\citep{turner_frechet_2013}}\label{lem:diff-Q}
\begin{enumerate}
\item $Q_D$ is Lipschitz continuous on any bounded set;
    \item The differential and gradient are well-defined for $Q_D$;
    \item If $D'$ is a local minimum of $Q_D$ then $\nabla_{D'}Q_D=0$.
\end{enumerate}
\end{lemma}
Since $\mathbf{F}$ is a linear combination of squared distance functions, by Lemma \ref{lem:diff-Q}, the gradient exists for any variable $Z_i$. Thus, we formally write the gradient of $\mathbf{F}$ as 
$
\nabla_{\mathbf{Z}}\mathbf{F} = (\nabla_{Z_1}\mathbf{F},\ldots,\nabla_{Z_k}\mathbf{F}).
$

Following \cite{selim1984k}, we define the KKT points for \eqref{eq:relax} as follows.
\begin{definition}
A point $(\omega_\bullet,\mathbf{Z}_\bullet)$ is a \emph{Karush--Kuhn--Tucker (KKT) point} of \eqref{eq:relax} if there exists $\mu_1,\ldots,\mu_n\in\mathbb{R}$ such that
$$
\begin{cases}
\displaystyle\frac{\partial \mathbf{F}}{\partial \omega_j}(\omega_\bullet,\mathbf{Z}_\bullet) + \mu_j \bm{1} = 0 & \forall\ j=1,\ldots,n\\
\bm{1}^\top(\omega_\bullet)_j - 1 = 0 & \forall\ j=1,\ldots,n\\
\nabla_\mathbf{Z}\mathbf{F}(\omega_\bullet,\mathbf{Z}_\bullet) = 0 &
\end{cases}
$$
where $\bm{1}$ is the all-1 vector.
\end{definition}

\begin{theorem}
The partial optimal points of \eqref{eq:relax} are KKT points of \eqref{eq:relax}.
\end{theorem}
\begin{proof}
Suppose $(\omega_\star,\mathbf{Z}_\star)$ is a partial optimal point. Since $\omega_\star$ is a minimizer of the function $\mathbf{F}(\cdot,\mathbf{Z}_\star)$, it satisfies the KKT conditions of the constrained optimization problem
%\begin{equation*}
$$
    \min_{\omega} \mathbf{F}(w,\mathbf{Z}_\star) \quad \text{s.t.} \quad \omega\in\Theta
$$
%\end{equation*}
which are exactly the first two conditions of being a KKT point. Similarly, since $\mathbf{Z}_\star$ is a minimizer of the function $\mathbf{F}(\omega_\star,\cdot)$, by Lemma \ref{lem:diff-Q}, the gradient vector is zero, which is the third condition of a KKT point. Thus, $(\omega_\star,\mathbf{Z}_\star)$ is a KKT point of \eqref{eq:relax}.
\end{proof}
%\begin{remark}
Conversely, suppose $(\omega_\star,\mathbf{Z}_\star)$ is a KKT point. Since $\mathbf{F}(\cdot,\mathbf{Z}_\star)$ is a linear function of $\omega$ and $\omega\in\Theta$ is a linear constraint, the first two conditions are sufficient for $\omega_\star$ to be a minimizer of $\mathbf{F}(\cdot,\mathbf{Z}_\star)$. Thus $\omega_\star$ satisfies the first condition of a partial optimal point. However, the condition $\nabla_{\mathbf{Z}}\mathbf{F}(\omega_\star,\mathbf{Z}_\star)=0$ cannot guarantee that $\mathbf{Z}_\star$ satisfies the second condition of a partial optimal point. Note that the original proof in \cite{selim1984k} of Theorem 4 is also incomplete.
%\end{remark}
%\begin{remark}

Moreover, being a partial optimal point or KKT point is not sufficient for being a local minimizer of the original optimization problem. A counterexample can be found in \cite{selim1984k}.
%\end{remark}

\subsection{Convergence to Local Minima}

We now give a characterization of the local minimality of $\omega$ using directional derivatives of the objective function $f(\omega)$ in \eqref{eq:reduced}. For any fixed element $\omega_0\in\Theta$, let $v\in\mathbb{R}^{kn}$ be a feasible direction, i.e., $\omega_0+tv\in\Theta$ for small $t$. The directional derivative of $f(\omega)$ along $v$ is defined as 
$$
Df(\omega_0;v):=\lim_{t\to 0^+}\frac{f(\omega_0+tv)-f(\omega_0)}{t},
$$
if the limit exists. Let $\mathcal{Z}(\omega_0)$ be the set of all $\mathbf{Z}$ minimizing the function $\mathbf{F}(\omega_0,\cdot)$. We have the following explicit formula for the directional derivatives of $f(\omega)$.
\begin{lemma}
For any $\omega_0\in\Theta$ and any feasible direction $v$, the directional derivative $Df(\omega_0;v)$ exists, and
$$
Df(\omega_0;v) = \min_{\mathbf{Z}\in\mathcal{Z}(\omega_0)}\Bigg\{\sum_{i=1}^k\sum_{j=1}^nv_{ij}\mathrm{W}_2^2(D_j,Z_i)\Bigg\}.
$$
\end{lemma}
\begin{proof}
Denote
$$
L(\omega_0;v) := \min_{\mathbf{Z}\in\mathcal{Z}(\omega_0)}\Bigg\{\sum_{i=1}^k\sum_{j=1}^nv_{ij}\mathrm{W}_2^2(D_j,Z_i)\Bigg\}.
$$
For any $\mathbf{Z}_\star\in\mathcal{Z}(\omega_0)$, we have
%\begin{equation*}
    \begin{align*}
    \frac{f(\omega_0+tv)-f(\omega_0)}{t}=&\frac{\min_{\mathbf{Z}\in\mathbf{S}^k}\{\mathbf{F}(\omega_0+tv,\mathbf{Z})\}-\min_{\mathbf{Z}\in\mathbf{S}^k}\{\mathbf{F}(\omega_0,\mathbf{Z})\}}{t}\\
    \le & \frac{\mathbf{F}(\omega_0+tv,\mathbf{Z}_\star)-\mathbf{F}(\omega_0,\mathbf{Z}_\star)}{t} =\sum_{i=1}^k\sum_{j=1}^nv_{ij}\mathrm{W}_2^2(D_j,Z_i).
    \end{align*}
%\end{equation*}
Thus we have 
\begin{equation}\label{eq:limsup}
\limsup_{t\to 0^+}\frac{f(\omega_0+tv)-f(\omega_0)}{t}\le L(\omega_0;v).
\end{equation}
It suffices to prove
$$
\liminf_{t\to 0^+}\frac{f(\omega_0+tv)-f(\omega_0)}{t}\ge L(\omega_0;v).
$$
To do this, let ${t_j}$ be a sequence tending to 0 and 
$$
\frac{f(\omega_0+t_jv)-f(\omega_0)}{t_j}\to \liminf_{t\to 0^+}\frac{f(\omega_0+tv)-f(\omega_0)}{t}
$$
as $j\to\infty$. Choose $\mathbf{Z}_j\in\mathcal{Z}(\omega_0+t_jv)$ for each $j$. Since $\mathbf{S}^k$ is relatively compact, there is a subsequence of $\{\mathbf{Z}_j\}$ converging to some point $\mathbf{Z}_\star$. We still denote the subsequence by $\{\mathbf{Z}_j\}$ for the sake of convenience. By Lemma \ref{lem:diff-Q}, $\mathbf{F}(\omega,\mathbf{Z})$ is continuous. Thus, for any $\mathbf{Z}\in\mathbf{S}^k$,
%\begin{align*}
$$
\mathbf{F}(\omega_0,\mathbf{Z}_\star) = \lim_{j\to\infty}\mathbf{F}(\omega_0+t_jv,\mathbf{Z}_j) \le  \lim_{j\to\infty} \mathbf{F}(\omega_0+t_jv,\mathbf{Z})=\mathbf{F}(\omega_0,\mathbf{Z}).
$$
%\end{align*}
That is, $\mathbf{Z}_\star\in\mathcal{Z}(\omega_0)$. Therefore, we have
%$$
\begin{align*}
&\liminf_{t\to 0^+}\frac{f(\omega_0+tv)-f(\omega_0)}{t} \\
=&\lim_{j\to\infty}\frac{\mathbf{F}(\omega_0+t_jv,\mathbf{Z}_j)-\mathbf{F}(\omega_0,\mathbf{Z}_\star)}{t_j}\\
=&\lim_{j\to\infty}\frac{\mathbf{F}(\omega_0+t_jv,\mathbf{Z}_j)-\mathbf{F}(\omega_0,\mathbf{Z}_j)}{t_j} +\lim_{j\to\infty}\frac{\mathbf{F}(\omega_0,\mathbf{Z}_j)   -\mathbf{F}(\omega_0,\mathbf{Z}_\star)}{t_j}\\
\ge &\lim_{j\to\infty}\frac{\mathbf{F}(\omega_0+t_jv,\mathbf{Z}_j)-\mathbf{F}(\omega_0,\mathbf{Z}_j)}{t_j}\\
\ge & L(\omega_0;v).
\end{align*}
%$$
Combining with \eqref{eq:limsup}, we showed that the limit exists and is equal to $L(\omega_0;v)$.
\end{proof}

Since the function $f(\omega)$ is a differentiable continuous function on a convex set, the first-order necessary condition of local minima is a standard result in optimization.
\begin{theorem}
If $\omega_\star$ is a local minimizer of $f(\omega)$ over $\Theta$, then for any feasible direction $v\in\mathbb{R}^{kn}$ at $\omega_\star$,
$$
Df(\omega_\star;v)\ge 0.
$$
\end{theorem}

%%%%%%%%%%%%%%%%%%%%%%%%%%%%%%%%%%%%%%%%%%%%%%%%%%%%%%%%%%%%%%%%%%%%
\section{Numerical Experiments}
\label{sec:exp}

We experimentally evaluate and compare performance of the various persistence diagram representations in the $k$-means algorithm on simulated data as well as a benchmark shape dataset.  We remark that our experimental setup differs from those conducted in the previous study by \cite{maroulas_kmeans_2017}, which compares persistence-based clustering to other non-persistence-based machine learning clustering algorithms.  Here, we are interested in the performance of the $k$-means algorithm in the context of persistent homology. We note that the aim of our study is not to establish superiority of persistence-based $k$-means clustering over other existing statistical or machine learning clustering methods, but rather to understand the behavior of the $k$-means clustering algorithm in settings of persistent homology.  The practical benefit to such a study is that it can make clustering by $k$-means actually feasible in certain settings where it would otherwise not be: a notable example that we will study later on Section \ref{sec:exp_pt_cloud} is on data comprising sets of point clouds, where the natural Gromov--Hausdorff metric between point clouds is difficult to compute, and moreover, where the Fr\'{e}chet mean for sets of point clouds---crucial for the implementation of $k$-means---has not yet been defined or studied, to the best of our knowledge. %This study is the first to systematically compare the PH representations and their respective $k$-means algorithms, in addition to testing on real data. 

Specifically, in our numerical experiments, we implement the $k$-means clustering algorithms on vectorized persistence diagrams, persistence diagrams, and persistence measures. In the spirit of machine learning, many embedding methods have been proposed for persistence diagrams so that existing statistical methods and machine learning algorithms can then be applied to these vectorized representations, as discussed earlier.  We systematically compare the performance of $k$-means on three of the most popular embeddings.  We also implement $k$-means on persistence diagrams themselves as well as persistence measures, which requires adapting the algorithm to comprise the respective metrics as well as the appropriate means.  In total, we study a total of five representations for persistent homology in the $k$-means algorithm: three are embeddings of persistence diagrams, persistence diagrams themselves, and finally, their generalization to persistence measures.

\subsection{Embedding Persistence Diagrams}

The algebraic topological construction of persistence diagrams as outlined in Section \ref{sec:ph} results in a highly nonlinear data structure where existing statistical and machine learning methods cannot be applied.  The problem of embedding persistence diagrams has been extensively researched in TDA, with many proposals for vectorizations of persistence diagrams.  In this paper, we study three of the most commonly-used persistence diagram embeddings and their performance in the $k$-means clustering algorithm: Betti curves; persistence landscapes \citep{bubenik_statistical_2015}; and persistence images \citep{adams_persistence_2017}. These vectorized persistence diagrams can then be directly applied to the original, non-persistence-based $k$-means algorithm. In this sense, we are providing a comparison of the implementation of the persistent homology version of $k$-means clustering to the implementation of classical $k$-means clustering on vectors computed from the output of persistent homology (persistence diagrams).  This experiment thus remains restricted to studying exclusively the $k$-means algorithm as a clustering method, as well as restricted to persistent homology, but nevertheless provides a comparison of the algorithm in two distinct metric spaces.

The \emph{Betti curve} of a persistence diagram $D$ is a function that takes any $z \in \mathbb{R}$ and returns the number of points $(x, y)$ in $D$, counted with multiplicity, such that $z \in [x, y)$. Figure \ref{fig:pd} shows an example of a persistence diagram and its corresponding Betti curve is shown in Figure \ref{fig:bc}. It is a simplification of the persistence diagram, where information on the persistence of the points is lost. The Betti curve takes the form of a vector by sampling values from the function at regular intervals. 

%Motivated to bridge the gap between PH and statistical methodology, \cite{bubenik_statistical_2015} introduced the PL as a functional summary of a PD without introducing loss of information. 

\begin{figure}
  \centering
  \subfloat[Persistence Diagram\label{fig:pd}]{\includegraphics[width=.4\textwidth,height=.4\textwidth]{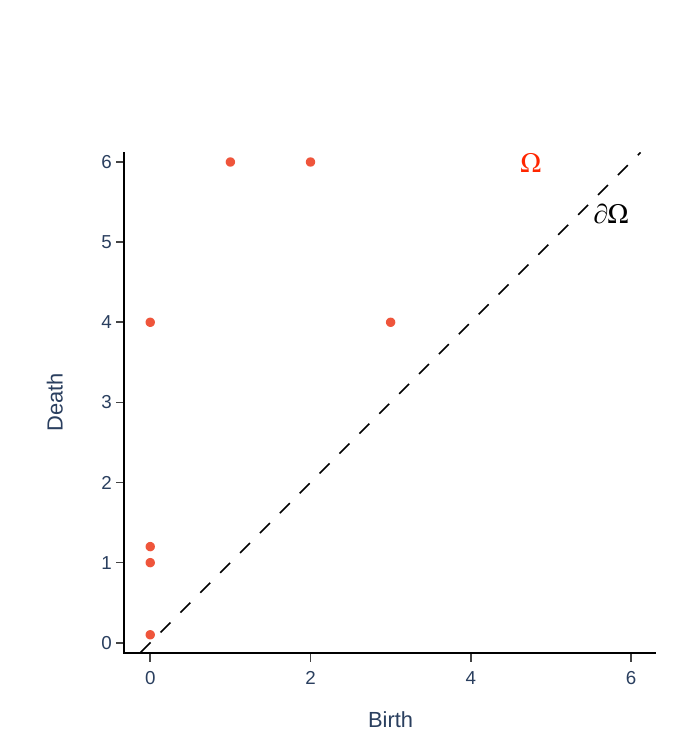}}
  \subfloat[Betti Curve\label{fig:bc}]{\includegraphics[width=.4\textwidth,height=.35\textwidth]{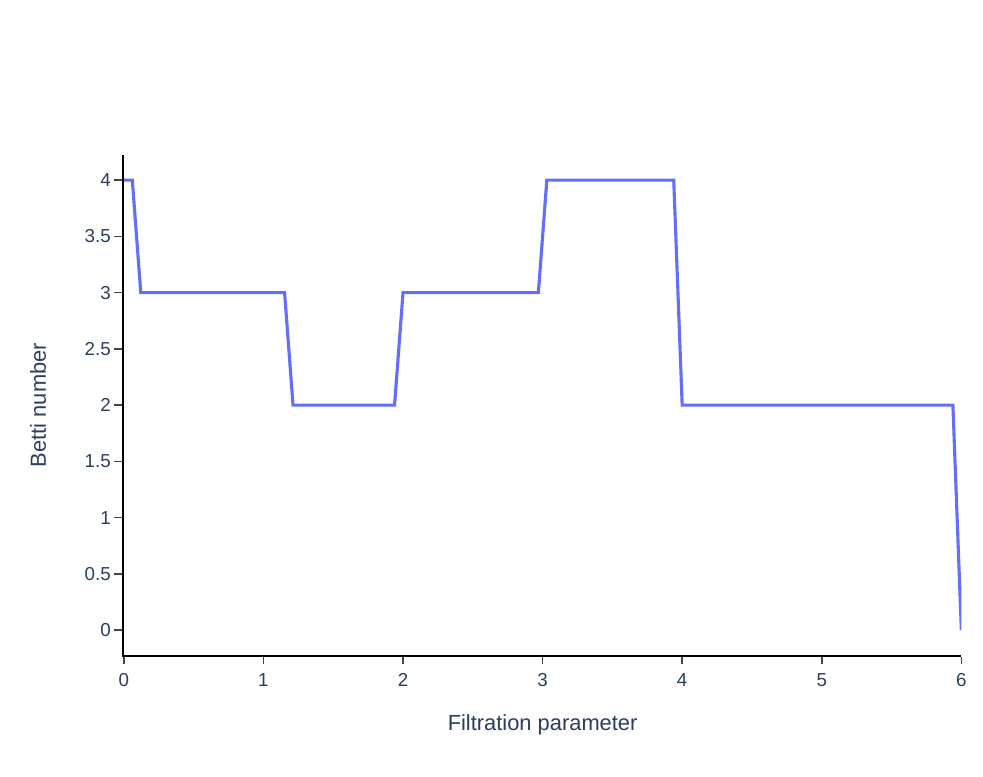}}\\
  \subfloat[Persistence Landscape\label{fig:pl}]{\includegraphics[width=.375\textwidth,height=0.35\textwidth]{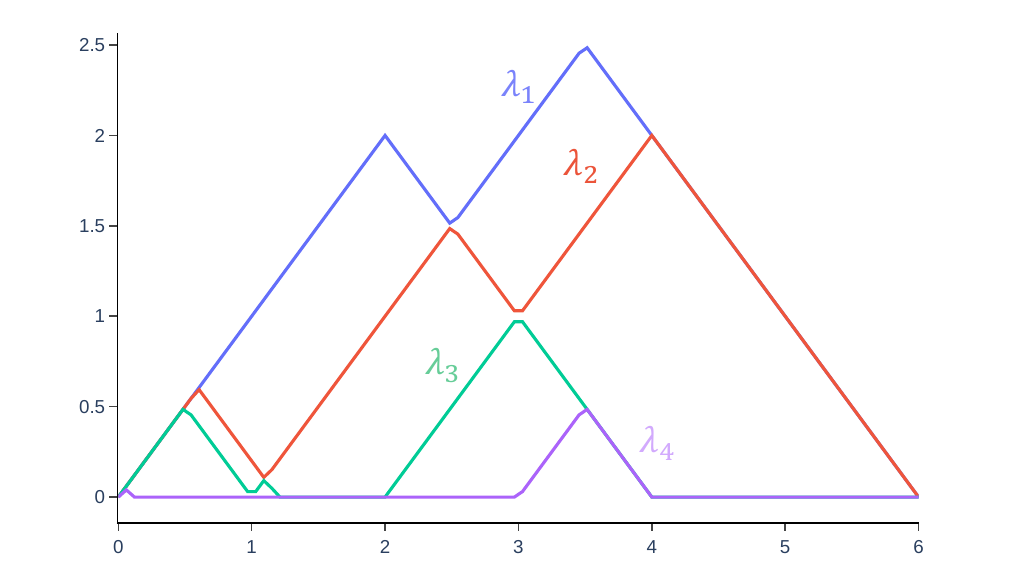}} 
  \subfloat[Persistence Image\label{fig:pi}]{\includegraphics[width=.375\textwidth,height=0.35\textwidth]{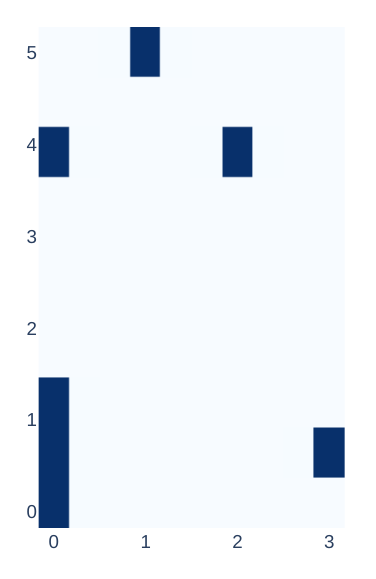}}
  \caption{Representations of Persistent Homology. \label{fig:pd_embed}}
\end{figure}

Persistence landscapes \citep{bubenik_statistical_2015} as functional summaries of persistence diagrams were introduced with the aim of integrability into statistical methodology.  The \emph{persistence landscape} is a set of functions $\{\lambda_k\}_{k \in \mathbb{N}}$ where $\lambda_k(t)$ is the $k$th largest value of 
$$
\max\bigg(0, \min_{(x, y) \in D}(t - x,\, y-t)\bigg).
$$
Figure \ref{fig:pl} continues the example of the persistence diagram in Figure \ref{fig:pd} and shows its persistence landscape. The construction of the persistence landscape vector is identical to that of the Betti curve, but for multiple $\lambda$ functions for one persistence diagram. 

The \emph{persistence image} \citep{adams_persistence_2017} is 2-dimensional representation of a persistence diagram as a collection of pixel values, shown in Figure \ref{fig:pi}. The persistence image is obtained from a persistence diagram by discretizing its \textit{persistence surface}, which is a weighted sum of probability distributions.  The original construction in \cite{adams_persistence_2017} uses Gaussians, which are also used in this paper. The pixel values are concatenated row by row to form a single vector representing one persistence diagram. 

The persistence image is the only representation we study that lies in Euclidean space; Betti curves and persistence landscapes are functions and therefore lie in function space.

We remark that there is a notable computational time difference between raw persistence diagrams and embedded persistence diagrams. This is because when updating centroids of persistence diagrams (i.e., mean persistence measures and Fr\'{e}chet means), expensive algorithms and techniques are involved to approximate the Wasserstein distance and search for local minima of the Fr\'{e}chet function \eqref{eq:frechet_function} \citep{turner_frechet_2013,flamary2021pot,lacombe_large_2018}.

%\begin{figure}
%\centering
%   \begin{subfigure}{\columnwidth}
%   \centering
%    \includegraphics[width=0.75\linewidth]{figs/persistence_diagram.pdf}
%    \caption{A persistence diagram}
%    \label{fig:PD}
%    \end{subfigure}
%    \begin{subfigure}[b]{0.65\columnwidth}
%    \includegraphics[width=\linewidth]{figs/betti_curve.pdf}
%    \caption{BC of Figure \ref{fig:PD}}
%    \label{fig:BC}
%    \end{subfigure}
%    \begin{subfigure}[b]{0.7\columnwidth}
%    \centering
%    \includegraphics[width=\linewidth]{figs/landscape.pdf}
%    \caption{PL of Figure \ref{fig:PD}}
%    \label{fig:PL}
%    \end{subfigure}
%    \begin{subfigure}{0.5\columnwidth}
%    \centering
%    \includegraphics[width=\linewidth]{figs/image.pdf}
%    \caption{PI of Figure \ref{fig:PD}}
%    \label{fig:PI}
%    \end{subfigure}
%\end{figure}

\subsection{Simulated Data}
\label{sec:exp_pt_cloud}

We generated datasets equipped with known labels, so any clustering output can be compared to these labels using the Adjusted Rand Index (ARI)---a performance metric where a score of 1 suggests perfect alignment of clustering output and true labels, and 0 suggests random clustering (\cite{hubert_comparing_1985}).

We simulated point clouds sampled from the surfaces of 3 classes of common topological shapes: the circle $S^1$, sphere $S^2$, and torus $T^2$. In terms of homology, any circle is characterized by one connected component and one loop, so we have $H_0(S^1) \cong \mathbb{Z}$ and $H_1(S^1) \cong \mathbb{Z}$.  Similarly, a sphere is characterized by one connected component and one void (bubble), so $H_0(S^2) \cong \mathbb{Z}$ and $H_2(S^2) \cong \mathbb{Z}$, but there are no loops to a sphere since every cycle traced on the surface of a sphere can be continuously deformed to a single point, so $H_1(S^2) \cong 0$.  Finally, for the case of the torus, we have one connected component so $H_0(T^2) \cong \mathbb{Z}$ and one void so $H_2(T^2) \cong \mathbb{Z}$; notice that there are two cycles on the torus that cannot be continuously deformed to a point, one that traces the surface to enclose the ``donut hole'' and the other smaller one that encircles the ``thickness of the donut'' so $H_1(T^2) \cong \mathbb{Z} \times \mathbb{Z}$.

We add noise from the uniform distribution on $[-s,s]$ coordinate-wise to each data point of the point clouds, where $s$ stands for the noise scale. Figure \ref{fig:noisy_torus} shows noisy point cloud data from the torus with different noise scales.

\begin{figure}[h]
\begin{subfigure}{0.32\textwidth}
    \centering
    \includegraphics[width=\textwidth]{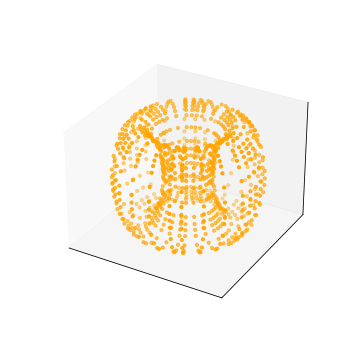}
\end{subfigure}
\begin{subfigure}{0.32\textwidth}
    \centering
    \includegraphics[width=\textwidth]{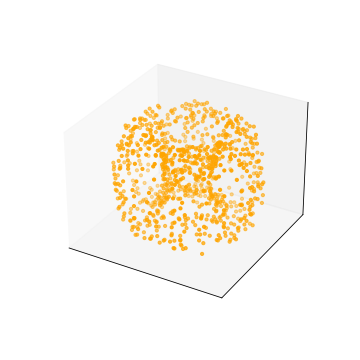}
\end{subfigure} 
\begin{subfigure}{0.32\textwidth}
    \centering
    \includegraphics[width=\textwidth]{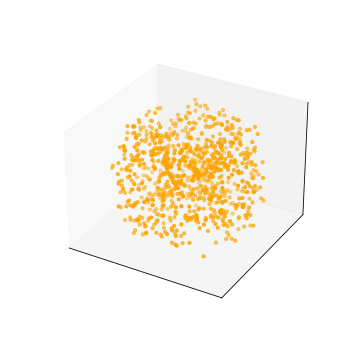}
\end{subfigure}
\caption{Torus point clouds with different level of noise $s=0,1,2$.}
\label{fig:noisy_torus}
\end{figure}
 
Table \ref{tab:synthetic_vector_results} shows the results of $k$-means clustering with 3 clusters on the five persistent homology representations of the simulated dataset. We fix the number of clusters to 3, given our prior knowledge of the data and because our focus is to consistently compare the performance accuracy of the $k$-means algorithms. Since circles, spheres, and tori are quite different in topology as described previously, the results from $k$-means algorithm is consistent with our knowledge of the topology of these classes of topological shapes. Figure \ref{fig:pd-synthetic} presents the persistence diagrams of the three different shapes we study. 

\begin{figure}[htbp]
\begin{subfigure}{0.32\textwidth}
    \centering
    \includegraphics[width=\textwidth]{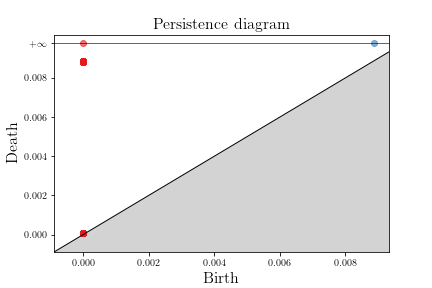}
\end{subfigure}
\begin{subfigure}{0.32\textwidth}
    \centering
    \includegraphics[width=\textwidth]{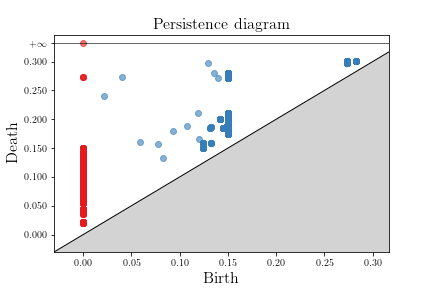}
\end{subfigure}
\begin{subfigure}{0.32\textwidth}
    \centering
    \includegraphics[width=\textwidth]{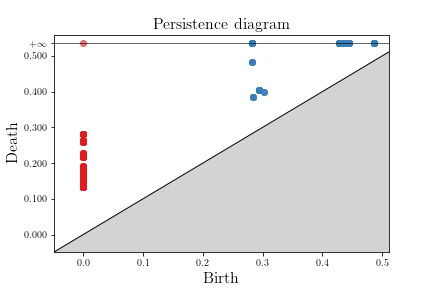}
\end{subfigure}
\caption{Sample persistence diagrams of a simulated circle, sphere, and torus.}
\label{fig:pd-synthetic}
\end{figure}

Tables \ref{tab:synthetic_vector_comparison} and \ref{tab:synthetic_repeat_persistence_results} show the average results from 100 repetitions comparing the persistence landscapes and persistence images, and the persistence diagrams and persistence measures, respectively. We see that persistence measures outperform persistence diagrams with consistently higher scores on average. Among the embedded persistence diagrams, the persistence landscapes produced the best clustering results, followed by the persistence images, and finally the Betti curves, which were not able to cluster accurately for any level of noise.  

\begin{table}[h]
\centering
\begin{subtable}[h]{0.65\textwidth}
\centering
\begin{tabularx}{\textwidth}{X|XXXXX}
  \hline
    Noise & PL & PI & BC & PD & PM \\
  \hline
    0.0 & \textbf{1.0} & \textbf{1.0} & \textbf{1.0} & \textbf{1.0} & \textbf{1.0} \\
    1.0 & \textbf{1.0} & 0.898 & 0.052 & \textbf{1.0} & \textbf{1.0} \\
    2.0 & 0.898 & \textbf{1.0} & 0.043 & \textbf{1.0} & \textbf{1.0}\\
    3.0 & 0.898 & 0.731 & 0.202 & \textbf{1.0} & \textbf{1.0}\\
    4.0 & \textbf{1.0} & \textbf{1.0} & 0.416 & \textbf{1.0} & \textbf{1.0}\\
    5.0 & 0.717 & 0.898 & 0.266 & \textbf{1.0} & \textbf{1.0}\\
    10.0 & 0.667 & 0.423 & 0.225 & \textbf{1.0} & 0.898\\
   \hline
\end{tabularx}
    \caption{Clustering results for simulated data}
    \label{tab:synthetic_vector_results}
    \end{subtable}

\begin{subtable}[h]{0.45\textwidth}
\centering
\begin{tabularx}{\textwidth}{X|XXX}
  \hline
    Noise & Mean PL score & Mean PI score\\
  \hline
    1.0 & \textbf{0.998} (0.012) & 0.892 (0.063) \\
    2.0 & \textbf{0.974} (0.037) & 0.965 (0.036) \\
    3.0 & 0.911 (0.082) & \textbf{0.922} (0.105) \\
    4.0 & \textbf{0.890} (0.085) & 0.852 (0.161) \\
    5.0 & \textbf{0.843} (0.112) & 0.734 (0.186) \\
    10.0 & \textbf{0.633} (0.088) & 0.453 (0.014) \\
   \hline
\end{tabularx}
    \caption{Performance of PLs vs PIs in 100 repetitions of the simulated data}
    \label{tab:synthetic_vector_comparison}
    \end{subtable}\hfill
\begin{subtable}[h]{0.45\textwidth}
\centering
\begin{tabularx}{\textwidth}{X|XX}
  \hline
    Noise & Mean PD score & Mean PM score \\
  \hline
    1.0 & 0.971 (0.202) & \textbf{0.996} (0.020) \\
    2.0 & 0.952 (0.260) & \textbf{0.961} (0.069) \\
    3.0 & 0.892 (0.251) & \textbf{0.947} (0.064) \\
    4.0 & 0.904 (0.218) & \textbf{0.946} (0.074) \\
    5.0 & 0.890 (0.220) & \textbf{0.945} (0.070) \\
    10.0 & 0.854 (0.185) & \textbf{0.892} (0.077) \\
   \hline
\end{tabularx}
    \caption{Clustering results for simulated data using PDs and PMs}
    \label{tab:synthetic_repeat_persistence_results}
    \end{subtable}
\end{table}

\subsection{3D Shape Matching Data}

We now demonstrate our framework on a benchmark dataset from 3D shape matching \citep{sumner2004deformation}. This dataset contains 8 classes: \texttt{Horse}, \texttt{Camel}, \texttt{Cat}, \texttt{Lion}, \texttt{Face}, \texttt{Head}, \texttt{Flamingo}, \texttt{Elephant}. Each class contains triangle meshes representing different poses. See Figure \ref{fig:shapes} for an visual illustration.

\begin{figure}[h]
    \centering
    \includegraphics[width=\linewidth]{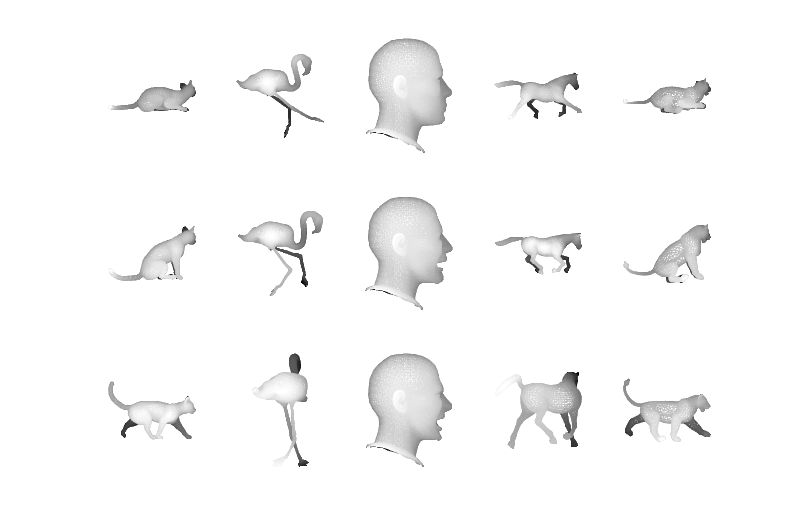}
    \caption{Samples from the shape matching database. Each column shows different poses of the same class.}
    \label{fig:shapes}
\end{figure}

For each triangle mesh we extract the coordinates of its vertices to obtain a point cloud. Then we compute the persistent homology of the point clouds. We apply the $k$-means clustering algorithm to raw/embedded persistence diagrams. From Table \ref{tab:shapes}, we see that the best ARI score appears when we set $k=3$. Moreover, we find that most point clouds in \texttt{Face} and \texttt{Head} form one cluster, \texttt{Horse} and \texttt{Camel} form the second, and \texttt{Cat} and \texttt{Lion} form the third, with other classes scattering in all these clusters. Our results are consistent with a similar clustering experiment previously performed by \cite{lacombe_large_2018}, who find two main clusters on a smaller database with 6 classes. We also find that persistence diagrams perform slightly better than persistence measures on this dataset. Among embedded persistence diagrams, the persistence landscapes give better clustering results, followed by persistence images and Betti curves, which is consistent with the performance seen in the previous section on the simulated datasets.    

\begin{table}[h]
    \centering
    \begin{tabularx}{\textwidth}{|X|XXXXX|}
    \hline
       Clusters  & PL & PI & BC & PD & PM \\
    \hline
       3  & 0.802 & 0.461 & 0.211 & \textbf{0.820} & 0.802\\
       4 & \textbf{0.471} & 0.398 & 0.158 & \textbf{0.471} & 0.465\\
       5 & 0.258 & \textbf{0.313} & 0.196 & 0.259 & 0.258\\
    \hline
    \end{tabularx}
    \caption{Clustering results for shape matching data.}
    \label{tab:shapes}
\end{table}

\subsection*{Software and Data Availability}
The code used to perform all experiments is publicly available at the GitHub repository: \url{https://github.com/pruskileung/PH-kmeans}.

%%%%%%%%%%%%%%%%%%%%%%%%%%%%%%%%%%%%%%%%%%%%%%%%%%%%%%%%%%%%%%%%%%%%

\section{Discussion}
\label{sec:end}

In this paper, we studied the $k$-means clustering algorithm in the context of persistent homology.  We studied the subtleties of the convergence of the $k$-means algorithm in persistence diagram space and in the KKT framework, which is a nontrivial problem given the highly nonlinear geometry of persistence diagram space.  Specifically, we showed that the solution to the optimization problem is a partial optimal point, KKT point, as well as a local minimum.  These results refine, generalize, and extend the existing study by \cite{maroulas_kmeans_2017}, which shows convergence to a partial optimal point, which need not be a local minimum.  Experimentally, we studied and compared the performance of the algorithm for inputs of three embedded persistence diagrams and modified the algorithm for inputs of persistence diagrams themselves as well as their generalizations to persistence measures.  We found that empirically, clustering results on persistence diagrams and persistence measures directly were better than on vectorized persistence diagrams, suggesting a loss of structural information in the most popular persistence diagram embeddings.  

Our results inspire new directions for future studies, such as other theoretical aspects of $k$-means clustering in the persistence diagram and persistence measure setting, for example, convergence to the ``correct'' clustering as the input persistence diagrams and persistence measures grow; and other properties of the algorithm such as convexity, other local or global optimizers, and analysis of the cost function.  The problem of information preservation in persistence embedding has been previously studied where statistically, an embedding based on tropical geometry results in no loss of statistical information via the concept of sufficiency \citep{doi:10.1137/17M1148037}.  Additional studies from the perspective of entropy would facilitate a better understanding and quantification of the information lost in embedding persistence diagrams that negatively affect the $k$-means algorithm.  This in turn would inspire more accurate persistence diagram embeddings for machine learning.  Another possible future direction for research is to adapt other clustering methods, including deep learning-based methods, to the setting of persistent homology. This would then allow for a comprehensive study of a wide variety of statistical and machine learning approaches for clustering in the context of persistent homology.

\section*{Acknowledgments}

We wish to thank Antonio Rieser for helpful discussions.  We also wish to acknowledge the Information and Communication Technologies resources at Imperial College London for their computing resources which were used to implement the experiments in this paper.

%\section{Tables}

\bibliographystyle{authordate3}
\bibliography{kmeans}

\end{document}